\documentclass[11pt]{article}
\usepackage{amsmath, amssymb, amsthm, fullpage, natbib}

\title{Sustainable Exploitation Equilibria \\for Dynamic Games with Irreversible Failure}
\author{
Nicholas H. Kirk$^{1,2}$\thanks{Email: \texttt{kirk@markovian.group}}\\
$^{1}$ Markovian Group\\
$^{2}$ Sa\"id Business School, University of Oxford
}
\date{}

\newtheorem{definition}{Definition}

\newtheorem{theorem}{Theorem}

\begin{document}
\maketitle

\begin{abstract}

We study dynamic relationships in which one party extracts current surplus in ways that degrade the future state, while the counterparty cannot exit but adjusts effort in response. Standard stationary Markov equilibria may sustain collapse paths in which short-run extraction dominates strictly positive continuation gains. We introduce the Sustainable Exploitation Equilibrium (SEE), a refinement for dynamic games with irreversible failure modeled as an absorbing boundary that eliminates continuation value. When a survival-preserving action exists and failure destroys future surplus, equilibria assigning positive probability to collapse are sequentially irrational. Equilibrium analysis can therefore be restricted, without loss, to continuation-preserving stationary Markov equilibria. Within this restricted domain, viable renegotiation-proofness becomes structural: because failure truncates future surplus, any jointly improving survival-preserving deviation is credible prior to collapse. SEE selects the viable, renegotiation-stable equilibrium that maximizes the exploiter’s value. Existence is established under standard conditions, and the refinement is illustrated in a hegemon–client setting.
\end{abstract}


\section{Introduction}
Exploitation arises when one party extracts value from another whose continued functioning is necessary for the relationship itself. In many economic and political environments, aggressive extraction can erode the very state that sustains the interaction—whether by depleting productive capacity, destabilizing institutions, or undermining cooperation. A central strategic tension therefore arises: the actor with extraction power benefits from pushing the system harder today, but doing so risks destroying the future surplus the relationship generates.

Standard equilibrium concepts such as Markov Perfect Equilibrium (MPE) do not fully resolve this tension. Because equilibria are evaluated over all admissible state transitions, they may sustain collapse paths in which short-run extraction dominates strictly positive continuation values. This paper introduces the Sustainable Exploitation Equilibrium (SEE), a refinement of stationary MPE for environments with irreversible failure. 

\paragraph{Why the refinement is necessary under irreversible failure.}
In dynamic environments without absorbing collapse, continuation payoffs are defined on a common domain across strategies. Irreversible failure changes this structure: once the state exits the survival region, the continuation problem collapses. Standard equilibrium notions evaluate strategies without restricting the admissible continuation domain, thereby permitting profiles that rely on the destruction of future surplus along the equilibrium path. SEE makes this restriction explicit by requiring that equilibrium paths preserve the continuation domain before selection among them. The refinement therefore reflects a consistency requirement imposed by irreversible failure rather than an additional constraint.

\section{Related Literature}
Our work connects three strands of existing literature.

\paragraph{Equilibrium refinements.}
SEE builds on the refinement tradition in dynamic games. \citep{MaskinTirole2001a} develop the theory of Markov perfect equilibrium; \citep{Bernheim1987} introduce coalition-proof Nash equilibrium; \citep{Selten1975} trembling-hand perfection; and \citep{FarrellMaskin1989} renegotiation-proofness in repeated games.

\paragraph{Viability and dynamic control.}
The viability constraint draws from \citep{Aubin1991} \emph{Viability Theory} and constrained Markov decision processes \citep{Altman1999}, ensuring state trajectories remain within sustainability sets. This connects to the resource economics literature, e.g. \citep{Dasgupta1979}, \citep{Chichilnisky1997}, and \citep{Dockner2000}, where long-run viability is central to dynamic exploitation of resources.

\paragraph{Political economy and relational contracts.}
SEE also relates to political economy models of extraction and state capacity (\citep{AcemogluRobinson2006}; \citep{BesleyPersson2011}) and to relational contracting without full commitment (\citep{ThomasWorrall1988}; \citep{Levin2003}). In these contexts, weaker parties cannot exit but retain effort margins that shape dynamics. SEE formalizes the logic of restrained exploitation within a Markovian framework.

\paragraph{Comparison.} In contrast to the above, this paper studies environments in which exit from a viability region constitutes irreversible failure that eliminates continuation value. Standard refinements select among equilibria defined on a fixed continuation domain. SEE instead restricts the admissible domain by requiring equilibrium paths to preserve the survival region before equilibrium selection. Viability and renegotiation-proofness therefore follow from sequential rationality under irreversible failure.

\section{Model and Definition}
\paragraph{Environment and timing.} Time is discrete, $t=0,1,2,\dots$. The state $s_t \in S$ evolves on a compact metric space $S$. In each period:
\begin{enumerate}
    \item The \emph{exploiter} $X$ observes $s_t$ and chooses extraction $x_t \in X(s_t)$.
    \item The \emph{exploitee} $E$ observes $(s_t,x_t)$ and chooses effort $e_t \in E(s_t)$.
    \item The state transitions to $s_{t+1} \sim Q(\cdot \mid s_t, x_t, e_t)$ (deterministic dynamics are $s_{t+1}=f(s_t,x_t,e_t)$).
\end{enumerate}
Per-period payoffs are $u^X(s_t,x_t,e_t)$ and $u^E(s_t,x_t,e_t)$, bounded and continuous. Both players discount with $\delta \in (0,1)$.

\subsection*{Roles of the players}
\begin{definition}[Exploiter]
The \emph{exploiter} $X$ directly extracts value via $x \in X(s)$, appropriating resources and typically reducing the future state. $X$'s stage payoff $u^X(s,x,e)$ is weakly increasing in $x$ for fixed $(s,e)$. A Markov policy for $X$ is $\sigma^X : S \to X(s)$.
\end{definition}

\begin{definition}[Exploitee (no exit, effort choice)]
The \emph{exploitee} $E$ cannot exit the relationship (no outside option) but chooses effort $e \in E(s)$ after observing $(s,x)$. Effort affects contemporaneous payoffs and/or the law of motion. A Markov policy for $E$ is $\sigma^E : S \times X(s) \to E(s)$ (second-mover response).
\end{definition}

\paragraph{Equilibrium notion (Markov--Stackelberg Equilibrium).}
A \emph{Markov--Stackelberg Equilibrium} (MSE) consists of policies $(\sigma^X,\sigma^E)$ such that:
(i) for each $(s,x)$, $\sigma^E(s,x)$ is a best response to $\sigma^X$ given follower timing;
(ii) $\sigma^X(s)$ anticipates $\sigma^E$ and is a best response given $E$'s continuation behavior.
Values $(W^X,W^E)$ solve the associated Bellman system.

\paragraph{Rational interpretation of viability.}
We interpret transitions outside a sustainability set $V \subset S$ as irreversible failure: absorption into a collapse region that eliminates continuation value. Under this interpretation, equilibrium behavior that places positive probability on exiting $V$ is incompatible with sequential rationality when failure entails sufficiently large losses and a survival-preserving action exists.

\begin{definition}[Sustainable Exploitation Equilibrium (SEE)]
Let $V\subseteq S$ be a nonempty compact viability set that is invariant under admissible state transitions. Let $\mathcal{E}$ denote the set of all stationary Markov Perfect Equilibria, and let $\mathcal{E}_V \subseteq \mathcal{E}$ denote the subset whose induced state process remains in $V$. Under rational viability (formalized below), restricting attention to $\mathcal{E}_V$ is without loss.

A \emph{Sustainable Exploitation Equilibrium} is any stationary MPE $\sigma^* \in \mathcal{E}$ satisfying:
\begin{enumerate}
    \item[(i)] \textbf{Viability:}
    The state induced by $\sigma^*$ remains in the viability set $V$ for all $s\in V$ (i.e., $\sigma^* \in \mathcal{E}_V$).
    \item[(ii)] \textbf{Renegotiation-proofness:}
    At any state $s\in V$, there is no alternative viable profile $\tilde{\sigma}\in\mathcal{E}_{V}$ such that
    $$ W_E(\tilde{\sigma}; s) \ge W_E(\sigma^*; s) \quad \text{and} \quad W_X(\tilde{\sigma}; s) \ge W_X(\sigma^*; s), $$
    with at least one inequality strict. Renegotiation is modeled in reduced form: because irreversible failure eliminates continuation value, any jointly improving survival-preserving deviation is credible prior to collapse. Renegotiation-proofness therefore acts as a stability requirement on continuation paths within the survival domain.
    \item[(iii)] \textbf{Exploiter-optimality:}
    Among all renegotiation-proof elements of $\mathcal{E}_{V}$, $\sigma^*$ maximizes the exploiter's value $W_X(\cdot; s)$ at state $s$ (ties broken lexicographically).
\end{enumerate}
Any profile satisfying (i)--(iii) is a Sustainable Exploitation Equilibrium.
\end{definition}

\paragraph{Multiplicity and uniqueness.}
The viability restriction removes collapse-inducing equilibria, while renegotiation-proofness eliminates survival-preserving equilibria that leave joint surplus unexploited. SEE therefore reduces the equilibrium correspondence relative to MPE and may yield uniqueness, although multiplicity can persist when several survival-preserving equilibria generate identical joint values.

\section{Relation to Nash and Markov Perfect Equilibrium}
We position the Sustainable Exploitation Equilibrium (SEE) within the standard equilibrium hierarchy.

\paragraph{Nash equilibrium (NE).}
A strategy profile $\sigma = (\sigma_X,\sigma_E)$ is a Nash equilibrium if no player can profitably deviate unilaterally. NE is broad and may rely on history-dependent strategies.

\paragraph{Subgame perfect equilibrium (SPE).}
A refinement of NE, SPE requires sequential rationality in every subgame. This eliminates non-credible threats. Clearly,
\[
\text{SPE} \subseteq \text{NE}.
\]

\paragraph{Markov perfect equilibrium (MPE).}
A further refinement of SPE for dynamic games with a payoff-relevant state $s \in S$. In MPE, strategies are Markovian: $\sigma_i(s)$ depends only on the current state, not the full history. For two-player sequential-move environments, the relevant benchmark is a \emph{Markov--Stackelberg equilibrium (MSE)}: the follower chooses $\sigma_E(s,x)$ as a best response given the state and leader’s action, while the leader anticipates this response in selecting $\sigma_X(s)$. Thus,
\[
\text{MPE (MSE)} \subseteq \text{SPE} \subseteq \text{NE}.
\]

\paragraph{Sustainable Exploitation Equilibrium (SEE).}
SEE is defined as a viable, renegotiation-proof MPE with exploiter-optimal selection. It follows immediately that
\[
\text{SEE} \subseteq \text{MPE} \subseteq \text{SPE} \subseteq \text{NE}.
\]

\paragraph{SEE is strictly stronger than renegotiation-proof MPE.}
Renegotiation-proofness selects among equilibria within a given continuation domain but does not restrict that domain itself. With an absorbing collapse state, equilibria may drive the system to failure when short-run gains exceed continuation values and yet remain renegotiation-proof because collapse eliminates future surplus. SEE excludes such outcomes by restricting attention to survival-preserving paths, thereby refining both equilibrium selection and the admissible continuation domain. Similar observations apply to coalition-proof refinements, which eliminate profitable coalitional deviations but likewise evaluate equilibria on a fixed continuation domain.

\section{Existence under rational viability}
We derive viability from rationality by interpreting exits from a sustainability set as irreversible failure.

\paragraph{Economic interpretation of irreversible failure.}

Irreversible failure is modeled as an absorbing boundary that eliminates continuation value. Economically this corresponds to events such as sovereign default, political collapse, financial liquidation, or irreversible depletion of productive capacity. When the discounted continuation loss exceeds the maximal short-run gain from over-extraction, survival-preserving behavior follows from sequential rationality.

\medskip
\noindent
\textbf{Assumption (Irreversible failure and safe action).}
Payoffs are bounded: $|u^X(s,x,e)| \le \bar u < \infty$. Let $V \subset S$ be a nonempty compact set. For each $s \in V$, there exists a \emph{safe} action $\underline{x}(s) \in X(s)$ such that, under the exploitee's best response,
\[
Q\!\left( V \mid s, \underline{x}(s), \sigma_E(s,\underline{x}(s)) \right) = 1 .
\]

\medskip
\noindent
\textbf{Assumption (Uniform exit risk for non-survival-preserving actions).}
For each $s \in V$, let
\[
\mathcal{X}_s^{\,N} \;=\; \{x \in X(s) : Q(V \mid s,x,\sigma_E(s,x)) < 1\}
\]
denote the set of actions that are not survival-preserving (given the exploitee's best response). Assume that the minimal probability of exiting $V$ over all such actions is bounded away from zero:
\[
p^\star \;=\; \inf_{s \in V} \inf_{x \in \mathcal{X}_s^{\,N}}
\big(1 - Q(V \mid s,x,\sigma_E(s,x))\big) \;>\; 0.
\]

\begin{theorem}[Existence of SEE under rational viability]
Suppose the state space $S$ and action sets $X(s), E(s)$ are compact, payoffs are bounded and continuous, and the transition kernel $Q$ satisfies the Feller property. Interpret $S \setminus V$ as an absorbing failure region that eliminates continuation value.

Then there exists $M^\star < \infty$ such that, for all failure penalties $M \ge M^\star$, every stationary Markov Perfect Equilibrium of the penalized game
\[
\hat u^X(s,x,e) \equiv u^X(s,x,e) - M \cdot \mathbf{1}\{ s' \notin V \},
\qquad s' \sim Q(\cdot \mid s,x,e),
\]
is viable. Moreover, among viable stationary Markov Perfect Equilibria, a nonempty compact subset of renegotiation-proof equilibria exists, and an exploiter-optimal selection exists. Hence a Sustainable Exploitation Equilibrium exists.
\end{theorem}

\begin{proof}
Under compactness and continuity, stationary Markov Perfect Equilibria exist. Fix any $s \in V$. If an equilibrium action induces $Q(S \setminus V \mid s,x,e) > 0$, the exploiter's expected one-period payoff is at most $\bar u - M p$ for some $p^*>0$. Under the safe action $\underline{x}(s)$, the expected penalty is zero and the payoff is at least $-\bar u$. For $M$ sufficiently large, the safe action strictly dominates any action with positive exit probability. By continuity and compactness, a finite threshold $M^\star$ exists. Hence any equilibrium must be viable. Existence of renegotiation-proof and exploiter-optimal selections follows by compactness and continuity.
\end{proof}

\paragraph{Finite collapse threshold.}
Because payoffs are bounded and safe actions exist at every state in the survival domain, there exists a finite failure threshold $M^*$ such that any action inducing a positive probability of exit is strictly dominated once $M \ge M^*$. The threshold depends on (i) the maximal one-period gain from over-extraction, (ii) the minimal probability of exit under non-safe actions, and (iii) the discount factor. In particular, higher exit risk under over-extraction lowers the required threshold, while greater patience reduces the attractiveness of failure paths.

\section{Example: Foreign Politics (Hegemon--Client SEE)}
\subsection{Primitives and timing}
Two players: a hegemon $H$ (the exploiter) and a client state $C$ (the exploitee). The state $s_t \in S \subset \mathbb{R}_+$ represents client \emph{political capacity/stability} (e.g., tax capacity, administrative control, consent). Time is discrete $t=0,1,2,\dots$.

Within each period:
\begin{enumerate}
    \item $H$ observes $s_t$ and chooses an \emph{extraction/pressure} level $x_t \in X(s_t)$ (e.g., tribute, resource concessions, policy compliance).
    \item $C$ observes $(s_t,x_t)$ and chooses \emph{effort} $e_t \in E(s_t)$ (e.g., governance effort, compliance/implementation, regime maintenance).
    \item The state updates to $s_{t+1} = f(s_t,e_t) - h(x_t)$, where $f$ is increasing in $s$ and $e$ and $h$ is increasing in $x$; or more generally $s_{t+1} \sim Q(\cdot \mid s_t,x_t,e_t)$.
\end{enumerate}
Both discount with factor $\delta \in (0,1)$. The client cannot exit the relationship (no outside option) but has strategic effort. A viability requirement $s_t \ge s_{\min}$ avoids regime collapse or state failure.

\subsection{Payoffs and feasibility}
Per-period payoffs are continuous and bounded:
\[
u^H(s,x,e) \;=\; \pi(x) \;-\; k(s,e),
\qquad
u^C(s,x,e) \;=\; b(s) \;-\; \phi(e) \;-\; d(x),
\]
with $\pi'(x)>0$, $k \ge 0$ (costs to monitor/coerce or instability spillovers), $b'(s)\ge 0$, $\phi$ convex with $\phi'(0)=0$, and $d'(x)\ge 0$. Feasibility (viability) imposes
\[
x \;\le\; f(s,e) - s_{\min}.
\]

\subsection{Equilibrium notion and SEE}
We adopt a Markov--Stackelberg Equilibrium (MSE): the follower $C$ best responds with $\sigma^C(s,x)$, and the leader $H$ best responds with $\sigma^H(s)$ anticipating $\sigma^C$. A \emph{Sustainable Exploitation Equilibrium} (SEE) is a viable, renegotiation-proof MSE with exploiter-optimal selection (as in Definition~1).

\subsection{Follower's best response}
Given $(s,x)$ and anticipating future play, $C$ solves
\[
\max_{e \in E(s)} \;\Big\{ b(s) - \phi(e) - d(x) + \delta W^C\!\big( f(s,e) - h(x) \big) \Big\}.
\]

\medskip
\noindent
\textbf{Remark.}
If $W_C$ is differentiable, the interior best response satisfies
\[
-\phi'(e^*) + \delta W_C'(s') f_e(s,e^*) = 0.
\]

\subsection{Leader's SEE problem}
Anticipating $e^*(s,x)$, the hegemon solves
\[
V^H(s) \;=\; \max_{x \in X(s)} \;\Big\{ \pi(x) - k(s,e^*(s,x)) + \delta\, V^H\!\big( f(s,e^*(s,x)) - h(x) \big) \Big\}
\]
subject to viability $x \le f(s,e^*(s,x)) - s_{\min}$. Let $\mu \ge 0$ be the multiplier on viability.

\section*{Additional Implementative Remarks}
\paragraph{Robustness to small outside options.}
The no-exit assumption sharpens the exploitation structure but is not essential. Suppose the exploitee has an outside option yielding value at most $\varepsilon>0$ relative to continuation within the survival domain. For sufficiently small $\varepsilon$, survival-preserving stationary Markov equilibria and the renegotiation-stable selection vary continuously. When collapse losses for the exploiter exceed the maximal short-run gain from over-extraction, the viability restriction continues to bind and SEE converges to the baseline no-exit refinement as $\varepsilon \to 0$.

\paragraph{Timing and simultaneity.}
The baseline adopts sequential timing in which extraction precedes effort adjustment, but the refinement does not depend on this structure. Under simultaneous moves, stationary Markov equilibria remain well-defined, and the survival-domain restriction continues to exclude collapse-inducing profiles. The mechanism therefore reflects the truncation of continuation value at failure rather than the timing protocol.

\paragraph{Empirical operationalization.}
The state $s_t$ can be proxied by observable fiscal or administrative capacity indices (e.g. World Bank Governance Indicators), extraction $x_t$ by documented concession or compliance flows, and $s_{\min}$ by sovereign default or regime-transition events observable in panel data. Structural parameters governing persistence and responsiveness can be calibrated using standard approaches in state-capacity models (e.g.\ \citet{BesleyPersson2011}). The viability threshold is disciplined by the present value of foregone extraction rents following collapse.

\section{Conclusion}

Dynamic relationships often involve asymmetric control over current surplus and future viability. When extraction degrades the state sustaining the interaction, standard equilibrium concepts may permit collapse paths even when continuation values remain positive.

We introduce the Sustainable Exploitation Equilibrium (SEE), a refinement of stationary Markov Perfect Equilibrium for environments with irreversible failure. Because failure eliminates continuation value, survival-preserving deviations are credible prior to collapse, making renegotiation-proofness endogenous and restricting equilibrium analysis to viable stationary equilibria. If unconstrained extraction preserves viability, SEE yields an interior outcome; if it would induce failure, equilibrium instead converges to the viability boundary—the maximal level of exploitation consistent with survival.

SEE thus provides a compact framework for analyzing persistent exploitative relationships while clarifying how irreversible failure disciplines incentives without eliminating pressures toward maximal sustainable exploitation.

\bibliographystyle{elsarticle-harv}
\bibliography{winwin}

\end{document}